\documentclass[sigconf,nonacm,screen]{acmart}
\usepackage{amsfonts}
\usepackage{amsmath}
\usepackage{amsthm}
\usepackage{array}
\usepackage{bm}
\usepackage{booktabs}
\usepackage{bxtexlogo}
\usepackage{caption}
\usepackage{enumitem}
\usepackage{float}
\usepackage{graphicx}
\graphicspath{{figs/}}
\usepackage[utf8]{inputenc}
\usepackage{makecell}
\usepackage{mathtools}
\usepackage{multirow}
\usepackage{natbib}
\usepackage{xcolor}
\setcitestyle{authoryear,numbers,aysep={}}
\usepackage{soul}
\usepackage{stfloats}
\usepackage{subcaption}
\usepackage[switch]{lineno}
\def\innerprod<#1>{\langle #1 \rangle}

\newtheorem{proposition}{Proposition}

\theoremstyle{definition}

\theoremstyle{remark}

\AtBeginDocument{%
  }

\begin{document}

\title{Feasible Multi-Asset Optimal Execution under Cash Constraints}

\author{Ryuji Hashimoto}
\affiliation{%
  \institution{Simudyne}
  \country{UK}
  }

\author{Namid R. Stillman}
\affiliation{%
  \institution{Simudyne}
  \country{UK}
  }
\email{namid@simudyne.com}

\renewcommand{\shortauthors}{Hashimoto et al.}
\definecolor{draftgreen}{RGB}{0,120,60}   
\newcommand{\gr}[1]{\textcolor{draftgreen}{#1}}

\begin{abstract}
Optimal execution (OE) in multi-asset settings involves complex interactions across assets, particularly through shared capital constraints during portfolio rebalancing. While existing models capture cross-impact and portfolio-level dynamics, they largely overlook the role of explicit cash constraints along the execution trajectory. As a result, the feasibility of execution strategies under limited capital remains poorly understood. In this paper, we extend the classical Almgren–Chriss framework to incorporate intertemporal constraints on expected cash consumption, requiring that the expected cumulative cash spent does not exceed a prescribed budget at every trading period. We show that the resulting multi-asset OE problem can be equivalently formulated as a quadratically constrained quadratic program (QCQP), and further establish that it admits a convex representation under mild conditions. This provides a tractable framework for analyzing execution strategies under dynamic capital constraints. Through controlled synthetic experiments, we show that the proposed cash constraints qualitatively alter OE schedules toward cash-feasible sell-first executions as the constraints become tighter. Furthermore, evaluations in an out-of-sample agent-based market simulator demonstrate that our method substantially reduces peak cash drawdown while maintaining implementation shortfall comparable to existing execution strategies. Our results highlight the importance of explicitly modeling financial feasibility in multi-asset execution and provide a foundation for bridging theoretical OE models with practical capital constraints.
\end{abstract}

\keywords{Multi-asset Optimal Execution, Portfolio Rebalancing, Cash Constraints}


\maketitle

\section{Introduction}
Optimal execution (OE) is a central problem in quantitative finance, aiming to execute large orders efficiently while accounting for market impact and price risk. A foundational framework is the Almgren--Chriss (AC) model~\citep{ac} which explicitly formulates the trade-off between execution cost and risk. Building on this framework, a large body of literature has extended OE models to incorporate more realistic market features~\citep{oe_order_types, oe_resiliency_lob1}. Recently, data-driven approaches have been proposed to design execution strategies directly from market data~\citep{oe_rl1,oe_rl2}. These developments have advanced the OE modeling toward more practical settings.

Another important line of research extends OE to the multi-asset setting, where an investor tries to execute a portfolio simultaneously. This problem is not a straightforward extension of the single-asset case. Prior work has emphasized asset interactions, including cross-impact~\citep{multi_asset_oe_cross_impact1} and covariance structures~\citep{multi_asset_oe_covariance1}, which give rise to qualitatively new behaviors absent in single-asset models.

However, existing studies on multi-asset OE largely overlook cash constraints, or treat them only heuristically without a tractable analytical formulation~\citep{oe_w_cash_constraint1}. By cash constraints, we refer to the requirement that cumulative cash consumption remain within available capital throughout execution, which we formulate in expectation, constraining the expected cash process at every trading period. This aspect is particularly critical in portfolio rebalancing settings that involve both buy and sell orders, where simultaneous execution across multiple assets induces interactions in cash flows. For example, if a large portion of buy orders is executed early in the trading horizon (see Figure~\ref{Fig:conceptual_diagram}), available cash may be depleted while the portfolio remains unbalanced, making the remaining execution plan both financially and risk-wise infeasible. Despite its practical importance, the conditions under which such a feasible execution trajectory exists have not been characterized.

\begin{figure*}[tbp]
  \centering
  \includegraphics[width=\linewidth]{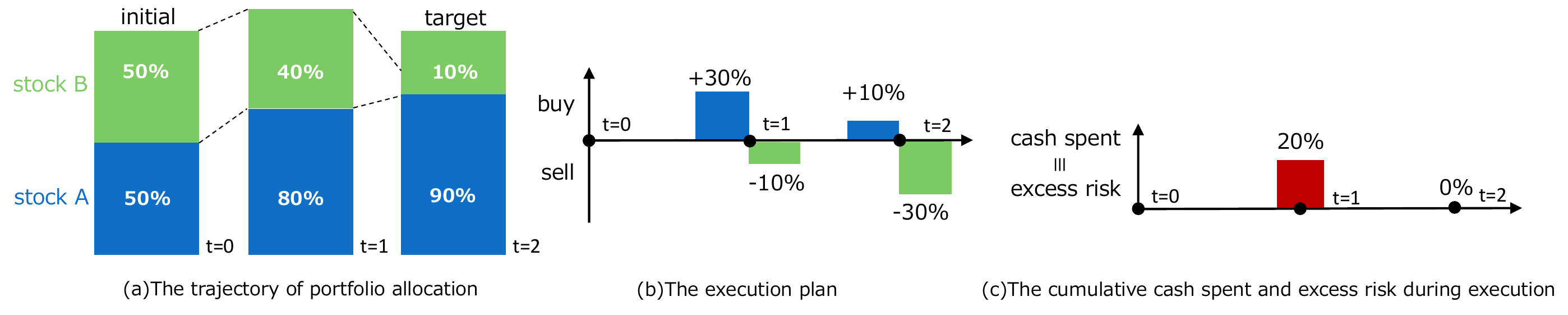}
  \caption{Conceptual illustration of the importance of cash constraints in multi-asset OE. Consider a stylized portfolio rebalancing problem with two assets over a two-period horizon as an example. Percentages denote each asset’s share of the total portfolio value at the initial time. We ignore differences in liquidity and volatility across assets. The goal is to rebalance the portfolio from an initial 50–50 allocation to a target allocation of 90\% in stock A and 10\% in stock B (a). We consider a buy-first execution strategy in which stock A is purchased aggressively in the first period (b), increasing its weight by 30\% while reducing stock B by 10\%. As a result, at the intermediate time ($t=1$), a cash outlay equivalent to 20\% of the initial portfolio value is required, which simultaneously induces an equivalent level of excess risk exposure (c).}
  \label{Fig:conceptual_diagram}
  \Description{Conceptual illustration of the importance of cash constraints in multi-asset OE.}
\end{figure*}

To address this gap, we extend the AC framework with explicit cash constraints imposed at every trading period, requiring that the expected cumulative cash spent up to each time step does not exceed a given budget. We prove that the resulting multi-asset OE problem can be equivalently formulated as a quadratically constrained quadratic program (QCQP). This reveals that, despite the path-dependent nature of cash feasibility, the problem remains within a tractable convex optimization class under mild conditions.

We validate the proposed formulation through two complementary experiments. First, in a controlled oracle setting where all market parameters are known, we demonstrate that tighter cash constraints progressively postpone buy orders while advancing sell orders, producing execution schedules that better align cash inflows with cash outflows. Second, using parameters calibrated from an agent-based market simulator, we evaluate the resulting execution strategies in out-of-sample simulations and show that they substantially reduce peak cash drawdown while maintaining implementation shortfall (IS) comparable to unconstrained OE and standard benchmark strategies. Taken together, these results establish our framework as a computationally tractable approach to safe multi-asset OE under explicit cash constraints, bridging classical OE theory and practical portfolio rebalancing under limited capital.

\section{Related Work}

OE aims to minimize the transaction costs of large orders, mainly arising from market impact and price risk. In practice, market impact is mitigated by order slicing. A seminal contribution is due to AC~\citep{ac}, which models temporary and permanent market impact linearly and formulates OE as a mean–variance optimization of IS, defined as the deviation between the pre-trade benchmark and realized execution cost, yielding an analytical solution.

While the AC framework assumes linear impact, subsequent work incorporates richer market dynamics, including order types~\citep{oe_order_types}, order flow~\citep{oe_order_flow}, dynamic impact~\citep{oe_dynamic_impact}, and liquidity resilience~\citep{oe_resiliency_lob1,oe_resiliency_lob2,oe_resiliency_lob3}. Recent studies have adopted data-driven approaches such as reinforcement learning~\citep{oe_rl1,oe_rl2,oe_rl4,oe_rl5,oe_rl6} to learn execution policies directly from data in complex market environments.

A separate line of work extends OE to multi-asset settings, where interactions across assets arise through cross-impact and portfolio-level considerations. Prior studies focus on modeling such interactions through covariance structures and stochastic resilience~\citep{multi_asset_oe_covariance1,multi_asset_oe_covariance2}, as well as cross-impact effects across assets~\citep{multi_asset_oe_cross_impact1,multi_asset_oe_cross_impact2}. While these models capture rich interdependencies in execution dynamics, they typically do not explicitly account for cash constraints over the execution horizon, despite their importance for ensuring feasibility. \citet{oe_w_cash_constraint1} is among the first to address this issue; however, their approach relies on data-driven methods, resulting in black-box policies without a tractable analytical formulation.

\section{Problem Setup}

\paragraph{\textbf{Multi-Asset Optimal Execution Settings}}
We consider a multi-asset optimal execution (OE) problem under explicit cash constraints. Let $N \in \mathbb{N}$ denote the number of assets and $T > 0$ the execution horizon. The trader aims to execute a portfolio of trades across these N assets within the finite time interval $[0, T]$. For each asset $n \in \{1,\ldots, N\}$, let $Q^{(n)} \in \mathbb{R}$ denote the total notional amount to be executed over the horizon. We adopt the convention that $Q^{(n)} > 0$ corresponds to a net purchase, while $Q^{(n)} < 0$ represents a net sale. Let $P_0^{(n)} \in \mathbb{R}_+$ denote the initial reference price of asset $n$. To model the execution process, we discretize the time interval $[0, T]$ into $K \in \mathbb{N}$ trading periods of equal length. Specifically, we define the discrete time grid as
\begin{align}
t_k \coloneqq k \tau, \quad k = 1, \ldots, K, \quad \tau = T / K.
\end{align}
At each time step $t_k$, the trader submits orders for each asset. Let $Q_k^{(n)} \in \mathbb{R}$ denote the execution amount of asset $n$ at time $t_k$.

\paragraph{\textbf{IS and Cash Spent}} We define the IS as the total execution cost relative to the initial reference prices. The IS is decomposed across assets as $IS\coloneqq\sum_{n=1}^N IS^{(n)}$, where, the IS for asset $n$ is:
\begin{align}
IS^{(n)}\coloneqq\sum_{k=1}^K \tilde{P}_k^{(n)} Q_k^{(n)} - P_0^{(n)} Q^{(n)}. \label{Eq:is}
\end{align}
Here, $\tilde{P}_k^{(n)}$ denotes the effective execution price of asset $n$ at time $t_k$, which incorporates both market impact and prevailing market conditions. Using $\tilde{P}_k^{(n)}$, the cumulative cash process $M_k\coloneqq\sum_{n=1}^NM_k^{(n)}$ is defined as the total monetary amount spent up to time $t_k$:
\begin{align}
M_k^{(n)}\coloneqq\sum_{j=1}^k\tilde{P}_j^{(n)}Q_j^{(n)}.\label{Eq:cash_spent}
\end{align}

\paragraph{\textbf{Price Process}} We assume that the asset price dynamics follow an AC-type discrete-time model. For each asset $n \in \{1, \ldots, N\}$ and each trading period $k = 1, \ldots, K$, the unaffected price $P_k^{(n)}$ and the effective execution price $\tilde{P}_k^{(n)}$ are given by
\begin{align}
\forall n 
\begin{cases}
P_k^{(n)}=P_{k-1}^{(n)}+\tau g^{(n)}\left(\frac{Q_k^{(n)}}{\tau}\right)+\sigma^{(n)}\sqrt{\tau}Z_k^{(n)},\\
\tilde{P}_k^{(n)}=P_k^{(n)}+h^{(n)}\left(\frac{Q_k^{(n)}}{\tau}\right)
\end{cases}.
\label{Eq:price_process}
\end{align}
Here, $g^{(n)}(\cdot)$ and $h^{(n)}(\cdot)$ represent the permanent and temporary market impact, respectively. $\sigma^{(n)} > 0$ is the volatility coefficient of asset $n$, and $\{Z_k^{(n)}\}_{k=1,\ldots,K,\;n=1,\ldots,N}$ are i.i.d. standard normal random variables, i.e., $Z_k^{(n)}\sim\mathcal{N}(0,1)$. In the following sections, we assume the both market impact functions are linear:
\begin{align}
g^{(n)}(x)=\gamma^{(n)}x, \qquad h^{(n)}(x)=\eta^{(n)}x,\label{Eq:linear_impact}
\end{align}
where $\gamma^{(n)},\eta^{(n)}\geq0$ denote the impact coefficients, respectively.

\paragraph{\textbf{Safe OE Problem}} We formulate the safe OE problem as a mean--variance optimization of the IS under cash constraints:
\begin{align}
\min_{\{Q_k^{(n)}\}}\mathbb{E}[IS]+\lambda Var(IS)\quad\text{s.t.}\quad\forall n~\sum_{k=1}^KQ_k^{(n)}=Q^{(n)},~~\forall k~\mathbb{E}[M_k]\leq \bar{c}_k.\label{Eq:oe}
\end{align}
Here, $\lambda \geq 0$ is a risk-aversion parameter. The operators $\mathbb{E}[\cdot], Var(\cdot)$ respectively denote expectation and variance with respect to the underlying probability space governing the price dynamics. The first constraint in Eq.~(\ref{Eq:oe}) enforces that the total executed amount for each asset matches the prescribed target $Q^{(n)}$. The second constraint imposes a dynamic cash limit, requiring that the expected cumulative cash process up to each time step $t_k$ does not exceed a given budget level $\bar{c}_k$. This formulation extends the classical AC framework by incorporating intertemporal budget constraints, which restrict not only the terminal outcome but also the path of capital deployment throughout the execution horizon.

\section{Optimal Execution as a QCQP}
\subsection{Overview of the Main Claim}
We reformulate the OE problem introduced in Eq.~(\ref{Eq:oe}) into a tractable optimization problem. In particular, our goal is to show that the problem can be expressed as a QCQP. This representation is key to analyzing the structural properties of the problem and enables the use of efficient convex optimization techniques. We further show that the resulting QCQP is convex, and hence can be solved efficiently to global optimality.

\begin{proposition}[QCQP reformulation of the safe OE problem]
Under the linear impact assumptions (Eq.~(\ref{Eq:linear_impact})), the OE problem Eq.~(\ref{Eq:oe}) can be equivalently reformulated as a QCQP. More precisely, there exist a decision vector $\bm{Q}\in\mathbb{R}^{NK}$, a symmetric matrix $\bm{H}\in\mathbb{R}^{NK\times NK}$, a matrix $\bm{C}\in\mathbb{R}^{N\times NK}$, a vector $\bm{b}\in\mathbb{R}^N$, and, for each $k=1,\ldots,K$, a vector $\bm{a}_k\in\mathbb{R}^{NK}$ and a symmetric matrix $\bar{\bm{D}}_k\in\mathbb{R}^{NK\times NK}$ such that the above OE problem is equivalent to
\begin{align}
\min_{\bm{Q}\in\mathbb{R}^{NK}} \bm{Q}^\top\bm{H}\bm{Q}
\quad\text{s.t.}\quad
\bm{C}\bm{Q}=\bm{b}, \quad \forall k~~\bm{a}_k^\top\bm{Q} + \bm{Q}^\top\bar{\bm{D}}_k \bm{Q} \leq \bar{c}_k.\label{Eq:oe_qcqp}
\end{align}
\end{proposition}

In the following subsections, we derive this QCQP representation in a constructive manner. Specifically, we express each component of the objective and constraints
in quadratic form.

\subsection{Derivation of the QCQP Reformulation}
In this subsection, we derive the QCQP reformulation in Eq.~(\ref{Eq:oe_qcqp}) through a sequence of algebraic transformations. In particular, we express the objective function, inventory constraint, and cash constraints as quadratic forms with respect to $\bm{Q}$, and explicitly construct the matrices $\bm H$, $\bm{C}$, and $\bar{\bm D}_k$ and vectors $\bm{b}$ and $\bm a_k$.

\paragraph{\textbf{Mean--Variance Objective}}
By denoting remaining orders to be executed before time $t_k$ as $q_k^{(n)}\coloneqq Q^{(n)}-\sum_{s=1}^{k-1}Q_s^{(n)}$, and rearranging the order of summation, $IS^{(n)}$ can be rewritten as follows:
\begin{align}
IS^{(n)}= \sum_{k=1}^K\left\{
\frac{\eta^{(n)}}{\tau}(Q_k^{(n)})^2+q_k^{(n)}\left(\gamma^{(n)}Q_k^{(n)}+\sigma^{(n)}Z_k^{(n)}\right)
\right\}.
\end{align}
Since $\{Z_k^{(n)}\}_{k,n}$ are i.i.d., the expectation and variance of $IS^{(n)}$ can be computed as follows:
\begin{align}
\mathbb{E}[IS^{(n)}]&= \sum_{k=1}^K\left\{
\frac{\eta^{(n)}}{\tau}(Q_k^{(n)})^2+\gamma^{(n)}q_k^{(n)}Q_k^{(n)}
\right\}.\label{Eq:e_isn}\\
Var(IS^{(n)})&= (\sigma^{(n)})^2\tau\,(\bm{Q}^{(n)})^\top\bm{U}_K^\top\bm{U}_K\bm{Q}^{(n)}.\label{Eq:var_isn}
\end{align}
Let $\bm{Q}^{(n)}\coloneqq\begin{pmatrix}Q_1^{(n)} & \ldots & Q_K^{(n)}\end{pmatrix}^\top\in\mathbb{R}^K,~ \bm{q}^{(n)}\coloneqq\begin{pmatrix}q_1^{(n)} & \ldots & q_K^{(n)}\end{pmatrix}^\top\in\mathbb{R}^K$ denote the execution schedule and remaining inventory vector for asset $n$. Using an upper triangular matrix: 
\begin{align}
\bm{U}_K\coloneqq\begin{pmatrix}
\bm{u}_1 & \ldots & \bm{u}_K
\end{pmatrix}\in\mathbb{R}^{K\times K},\quad
\forall k~ \bm{u}_k\in\mathbb{R}^K,~~(\bm{u}_k)_i\coloneqq\begin{cases}1\quad i\leq k\\ 0\quad k<i\end{cases},
\end{align}
we obtain: $\bm{q}^{(n)}=\bm{U}_K\bm{Q}^{(n)}$. Using this relation, the expectation and variance of $IS^{(n)}$ (Eqs.(\ref{Eq:e_isn},\ref{Eq:var_isn})) can be expressed in quadratic form with respect to $\bm{Q}^{(n)}$:
\begin{align}
\mathbb{E}[IS^{(n)}]&=\frac{\eta^{(n)}}{\tau}(\bm{Q}^{(n)})^\top\ \bm{Q}^{(n)}+\gamma^{(n)}(\bm{U}_K\bm{Q}^{(n)})^\top\ \bm{Q}^{(n)}\\
&= (\bm{Q}^{(n)})^\top\left(\frac{\eta^{(n)}}{\tau}\bm{E}_K+\frac{\gamma^{(n)}}{2}(\bm{U}_K+\bm{U}_K^\top)\right)\bm{Q}^{(n)}.\\
Var(IS^{(n)})&= (\sigma^{(n)})^2\tau{}^\top\bm{Q}^{(n)}{}^\top\bm{U}_K\ \bm{U}_K\bm{Q}^{(n)}.
\end{align}
Here, $\bm{E}_K\in\mathbb{R}^{K\times K}$ denotes the identity matrix. Let $\bm{Q}\in\mathbb{R}^{NK}$ denotes the overall execution plans: $\bm{Q}\coloneqq\begin{pmatrix}(\bm{Q}^{(1)})^\top & \ldots & (\bm{Q}^{(N)})^\top\end{pmatrix}^\top$. Summing over all assets, the expectation and variance of the total IS are given as follows.
\begin{align}
\mathbb{E}[IS] &= \bm{Q}^\top\bm{A}\bm{Q},\quad
Var(IS) = \bm{Q}^\top\bm{B}\bm{Q},\\
\bm{A} &\coloneqq \text{blkdiag}(\bm{A}^{(n)}),~
\bm{A}^{(n)}\coloneqq\frac{\eta^{(n)}}{\tau}\bm{E}_K+\frac{\gamma^{(n)}}{2}(\bm{U}_K+\bm{U}_K^\top),\label{Eq:A}\\
\bm{B} &\coloneqq \text{blkdiag}(\bm{B}^{(n)}),~
\bm{B}^{(n)}\coloneqq(\sigma^{(n)})^2\tau\,\bm{U}_K^\top\bm{U}_K.\label{Eq:B}
\end{align}
Here, $\text{blkdiag}(\bm{A}^{(n)})$ denotes the block diagonal matrix with blocks $\bm{A}^{(1)},\ldots,\bm{A}^{(N)}$. Combining the above results, the objective function can be written in quadratic form as
\begin{align}
\mathbb{E}[IS] + \lambda\, Var(IS)
= \bm{Q}^\top(\bm{A} + \lambda \bm{B})\bm{Q}.
\end{align}
Thus, by defining $\bm{H} \coloneqq \bm{A} + \lambda \bm{B}$, we obtain the desired quadratic representation of the objective.

\paragraph{\textbf{Inventory Constraint}} The inventory constraints admit a simple linear representation. Let $\bm{C}\in\mathbb{R}^{N\times NK}$ be the matrix whose $n$-th row extracts the total executed amount of asset $n$, i.e.,
\begin{align}
\bm{C} \coloneqq \text{blkdiag}\big(\bm{1}_K^\top\big),
\end{align}
where $\bm{1}_K^\top \in \mathbb{R}^{1\times K}$ denotes the row vector of ones. Let $\bm{b}\in\mathbb{R}^N$ be the vector of target execution amounts, $\bm{b} \coloneqq \begin{pmatrix} Q^{(1)} & \cdots & Q^{(N)} \end{pmatrix}^\top$. Then, the inventory constraints can be written as $\bm{C}\bm{Q} = \bm{b}$.

\paragraph{\textbf{Cash Constraint}}
Using Eq.~(\ref{Eq:price_process}), the cumulative cash spent $M_k^{(n)}$ (Eq.~(\ref{Eq:cash_spent})) can be expressed explicitly as follows:
\begin{align}
M_k^{(n)}&=P_0^{(n)}\sum_{j=1}^kQ_j^{(n)}+\gamma^{(n)}\sum_{j=1}^kQ_j^{(n)}\sum_{l=1}^jQ_l^{(n)}\nonumber\\
&~ +\frac{\eta^{(n)}}{\tau}\sum_{j=1}^k\left(Q_j^{(n)}\right)^2+\sigma^{(n)}\sqrt{\tau}\sum_{j=1}^kQ_j^{(n)}\sum_{l=1}^jZ_l^{(n)}.\label{Eq:cash_spent_explicit}
\end{align}
For Eq.~(\ref{Eq:cash_spent_explicit}), the first term can be written in inner product form:
\begin{align}
\sum_{j=1}^kQ_j^{(n)}=\bm{u}_k^\top\ \bm{Q}^{(n)}.\label{Eq:Q_till_k}
\end{align}
The second term in Eq.~(\ref{Eq:cash_spent_explicit}) can be written in quadratic form by introducing the matrix $\bm{L}_k\in\mathbb{R}^{K\times K}$ defined by
\begin{align}
(\bm{L}_k)_{i,j}\coloneqq
\begin{cases}
1 & 1\leq j\leq i\leq k,\\
0 & \text{otherwise}.
\end{cases}
\end{align}
Then,
\begin{align}
\sum_{j=1}^k Q_j^{(n)} \sum_{l=1}^j Q_l^{(n)}=
(\bm{Q}^{(n)})^\top\ \bm{L}_k \bm{Q}^{(n)}=
(\bm{Q}^{(n)})^\top\ \frac{\bm{L}_k+\bm{L}_k^\top}{2}\bm{Q}^{(n)}.\label{Eq:Qj_till_k}
\end{align}
The third term in Eq.~(\ref{Eq:cash_spent_explicit}) can also be written in quadratic form by introducing the matrix $\bm{D}_k\coloneqq\text{diag}(\bm{u}_k)\in\mathbb{R}^{K\times K}$:
\begin{align}
\sum_{j=1}^k\left(Q_j^{(n)}\right)^2=(\bm{Q}^{(n)})^\top\ \bm{D}_k \bm{Q}^{(n)},\label{Eq:QQ_till_k}
\end{align}
where $\text{diag}(\bm{u}_k)$ denotes the diagonal matrix whose diagonal entries are given by the components of $\bm{u}_k$. Using Eqs.~(\ref{Eq:Q_till_k},\ref{Eq:Qj_till_k},\ref{Eq:QQ_till_k}), the expected $M_k^{(n)}$ can be expressed as follows.
\begin{align}
\mathbb{E}[M_k^{(n)}]=P_0^{(n)}\bm{u}_k^\top\bm{Q}^{(n)}+(\bm{Q}^{(n)})^\top
\left(\frac{\gamma^{(n)}}{2}(\bm{L}_k+\bm{L}_k^\top)+\frac{\eta^{(n)}}{\tau}\bm{D}_k\right)
\bm{Q}^{(n)}
\end{align}
Therefore, by defining the vector $\bm{a}_k\in\mathbb{R}^{NK}$ and the matrix $\bar{\bm{D}}_k\in\mathbb{R}^{NK\times NK}$ as follows, we obtain $\forall k~~\mathbb{E}[M_k]= \bm{a}_k^\top\bm{Q} + \bm{Q}^\top\bar{\bm{D}}_k \bm{Q}$.
\begin{align}
\bm{a}_k&\coloneqq \begin{pmatrix}
P_0^{(1)}\bm{u}_k^\top & \ldots & P_0^{(N)}\bm{u}_k^\top
\end{pmatrix}^\top,\label{Eq:a}\\
\bar{\bm{D}}_k&\coloneqq \text{blkdiag}(\bar{\bm{D}}_k^{(n)}),~~\bar{\bm{D}}_k^{(n)}\coloneqq\left(
\frac{\gamma^{(n)}}{2}(\bm{L}_k+\bm{L}_k^\top)+\frac{\eta^{(n)}}{\tau}\bm{D}_k
\right)\label{Eq:D}
\end{align}

\begin{proposition}[Convexity of the QCQP]
Under $\lambda\geq0$, $\gamma^{(n)}\geq0$, and $\eta^{(n)}\geq0$ for all $n$, 
the QCQP formulation in Eq.~(\ref{Eq:oe_qcqp}) is convex.
\end{proposition}
\begin{proof}
The objective function is quadratic with Hessian $\bm{H}=\bm{A}+\lambda\bm{B}$. We first show that $\bm{H}\succeq 0$. For each asset $n$, the matrix $\bm{B}^{(n)}$ is given by Eq.~(\ref{Eq:B}). Since $\bm{U}_K^\top\ \bm{U}_K \succeq 0$ and $(\sigma^{(n)})^2 \tau \geq 0$, it follows that $\forall n~\bm{B}^{(n)} \succeq 0$, and thus $\bm{B} \succeq 0$. Next, consider $\bm{A}^{(n)}$ defined in Eq~(\ref{Eq:A}). Since $\bm{E}_K \succeq 0$ and $\eta^{(n)} \ge 0$, the first term in Eq~(\ref{Eq:A}) is positive semidefinite. 
Moreover, for any $\bm{x} \in \mathbb{R}^K$,
\begin{align}
\bm{x}^\top\ \bm{U}_K \bm{x}
= \sum_{1 \le i \le j \le K} x_i x_j
= \frac{1}{2} \left( \left( \sum_{i=1}^K x_i \right)^2 + \sum_{i=1}^K x_i^2 \right) \ge 0,
\end{align}
which implies $(\bm{U}_K + \bm{U}_K^\top)/2 \succeq 0$. 
Since $\gamma^{(n)} \ge 0$, we obtain $\forall n~\bm{A}^{(n)} \succeq 0$, and hence $\bm{A} \succeq 0$. Combining the above and using $\lambda \geq 0$, we obtain $\bm{H} = \bm{A} + \lambda \bm{B} \succeq 0$, so the objective function is convex.

Each equality constraint is affine, and therefore defines a convex feasible set. Each inequality constraint is of the form $\bm{a}_k^\top \bm{Q} + \bm{Q}^\top \bar{\bm{D}}_k \bm{Q} \leq \bar{c}_k$. For each $k$ and $n$, the matrix $\bar{\bm{D}}_k^{(n)}$ is given in Eq.~(\ref{Eq:D}). Since $\bm{D}_k = \text{diag}(\bm{u}_k) \succeq 0$ and $\eta^{(n)} \geq 0$, the second term of $\bar{\bm{D}}_k^{(n)}$ is positive semidefinite. Similarly, for any $\bm{x} \in \mathbb{R}^K$,
\begin{align}
\bm{x}^\top \bm{L}_k \bm{x}
= \sum_{1 \le l \le j \le k} x_j x_l
= \frac{1}{2} \left( \left( \sum_{j=1}^k x_j \right)^2 + \sum_{j=1}^k x_j^2 \right) \ge 0,
\end{align}
which implies $(\bm{L}_k + \bm{L}_k^\top)/2 \succeq 0$. 
Since $\gamma^{(n)} \ge 0$, we obtain $\bar{\bm{D}}_k^{(n)} \succeq 0$ for all $n$, and hence $\bar{\bm{D}}_k \succeq 0$.

Therefore, all inequality constraints are convex. This proves that the problem is a convex QCQP.
\end{proof}

\section{Experimental Setup}\label{Sec:experiment}
We evaluate the proposed framework through two complementary experiments that separately investigate its underlying mechanism and practical effectiveness.

\begin{itemize}[
    leftmargin=1.em,
    itemsep=0.4ex,
    topsep=0.6ex,
    parsep=0pt,
    partopsep=0pt
]
    \item \textbf{RQ1 (Mechanism Validation):} How do explicit cash constraints reshape the OE schedule?
    \item \textbf{RQ2 (Performance Evaluation):} Can the proposed formulation improve cash feasibility without sacrificing execution efficiency under realistic market conditions?
\end{itemize}
The first experiment addresses RQ1 using an oracle setting in which all market parameters are assumed to be known. This enables us to isolate the qualitative effect of the proposed cash constraints on the OE schedule. The second experiment addresses RQ2 by evaluating the optimized execution strategies in an out-of-sample agent-based market simulator with calibrated market parameters, thereby assessing their practical effectiveness.

\subsection{Experiment 1: Mechanism Validation} To answer \textbf{RQ1}, we conduct a controlled experiment to isolate the qualitative effect of the proposed cash constraints on the OE schedule. Since our objective is to understand the qualitative behavior induced by the proposed QCQP formulation itself, we assume an oracle setting in which all market parameters are known exactly. We consider a stylized portfolio rebalancing task involving two assets, where one asset is purchased and the other is sold simultaneously. The cash constraint is parameterized as $\bar{c}_k=\bar{c}_0+\beta k$, where $\beta$ controls the rate at which the available budget is relaxed over the execution horizon. By varying $\beta$, we investigate how different levels of cash availability influence the resulting OE schedules. The complete experimental settings are summarized in Table~\ref{Tab:exp1_setting}.

\begin{table}[t]
\centering
\caption{Simulation settings for Experiment 1.}
\label{Tab:exp1_setting}
\begin{tabular}{lll}
\toprule
\textbf{Param.} & \textbf{Description} & \textbf{Value} \\
\midrule
$N$ & \# of assets & $2$ \\
$Q^{(1)}$ & Buy-side order size & $20$ \\
$Q^{(2)}$ & Sell-side order size & $-20$ \\
$P_0^{(n)}$ & Initial asset price & $50.0,~ n=1,2$ \\
$\gamma^{(n)}$ & Permanent impact coefficient & $0.02,~ n=1,2$ \\
$\eta^{(n)}$ & Temporary impact coefficient & $0.05,~ n=1,2$ \\
$\sigma^{(n)}$ & Volatility & $0.10,~ n=1,2$ \\
$K$ & \# of execution periods & $20$ \\
$\tau$ & Time interval & $1$ \\
$\lambda$ & Risk-aversion parameter & $0.30$ \\
$\bar{c}_0$ & Initial cash constraint & $5.0$ \\
$\beta$ & Cash constraint relaxation rate & $0.1,\ 0.2,\ 0.3,\ 0.5$ \\
$-$ & Optimization solver & SCS \\
\bottomrule
\end{tabular}
\end{table}

\subsection{Experiment 2: Performance Evaluation}

To answer \textbf{RQ2}, we evaluate the proposed method in a realistic market environment using the \textit{Simudyne Pulse} agent-based market simulator~\citep{simudyne2025pulse}. Unlike Experiment~1, where all market parameters are assumed to be known, this experiment considers the practical setting in which the market impact and volatility parameters must be estimated from historical data. We therefore construct execution schedules using estimated market parameters and evaluate their performance under out-of-sample market conditions.

\paragraph{\textbf{Evaluation Protocol}}

To ensure a rigorous out-of-sample evaluation, parameter estimation and performance evaluation are conducted on disjoint trading days. Market parameters are calibrated using simulation data generated on the calibration date, while execution performance is evaluated on an independent validation date. The calibrated parameters are fixed throughout the evaluation and are never updated using the validation data. 

To evaluate the resulting strategies under unseen market conditions, we calibrate a separate Simudyne Pulse market simulator to the validation date (2025-10-14) and generate independent Monte Carlo price trajectories. The execution-model parameters are therefore estimated exclusively from the calibration market, while the execution strategies are evaluated in a distinct market environment calibrated to the validation date, ensuring a strict out-of-sample evaluation protocol.

\paragraph{\textbf{Parameter Estimation}}
For each asset $n$, the arrival price $P_0^{(n)}$, permanent market impact coefficient $\gamma^{(n)}$, temporary market impact coefficient $\eta^{(n)}$, and volatility $\sigma^{(n)}$ are estimated from the calibration data generated by the market simulator. To estimate the market impact parameters, we perform a series of TWAP executions with different target execution sizes and record the resulting IS. Let $V_i^{(n)}$ denote the total executed quantity in calibration run $i$, and let $IS_i^{(n)}$ denote the corresponding realized IS. Under the linear impact assumption, the expected IS satisfies

\begin{align}
\frac{IS_i^{(n)}}{V_i^{(n)}}=\eta^{(n)}+\frac{\gamma^{(n)}}{2}V_i^{(n)}+
\varepsilon_i,
\end{align}
where $\varepsilon_i$ denotes the residual error. The coefficients $\gamma^{(n)}$ and $\eta^{(n)}$ are estimated by ordinary least squares (OLS). The volatility parameter $\sigma^{(n)}$ is estimated from the variance of one-period mid-price returns obtained from baseline simulations without algorithmic execution. The arrival price $P_0^{(n)}$ is defined as the observed market price at the beginning of the validation period.

\paragraph{\textbf{Evaluation Settings}} We compare the proposed cash-constrained optimization (\textbf{Ours}) with three benchmark strategies: time-weighted average price (TWAP), volume-weighted average price (VWAP), and an ablated variant without the proposed cash constraints (\textbf{AC free}), which corresponds to Eq.~(\ref{Eq:oe}) with the cash constraints removed. AC free serves as an ablation of the proposed formulation, allowing us to isolate the contribution of the proposed cash constraints while keeping the optimization objective and all other settings identical. The remaining experimental settings are summarized in Table~\ref{Tab:exp2_setting}.

\paragraph{\textbf{Evaluation Metrics}} We report two evaluation metrics: the \emph{peak cash drawdown}, defined as the maximum cumulative cash deficit during execution normalized by the target buy notional, which measures cash feasibility throughout the execution process, and \emph{combined IS}, computed over the buy--sell portfolio, which evaluates the overall execution efficiency.

\begin{figure}[t]
    \centering
    \includegraphics[width=0.9\linewidth]{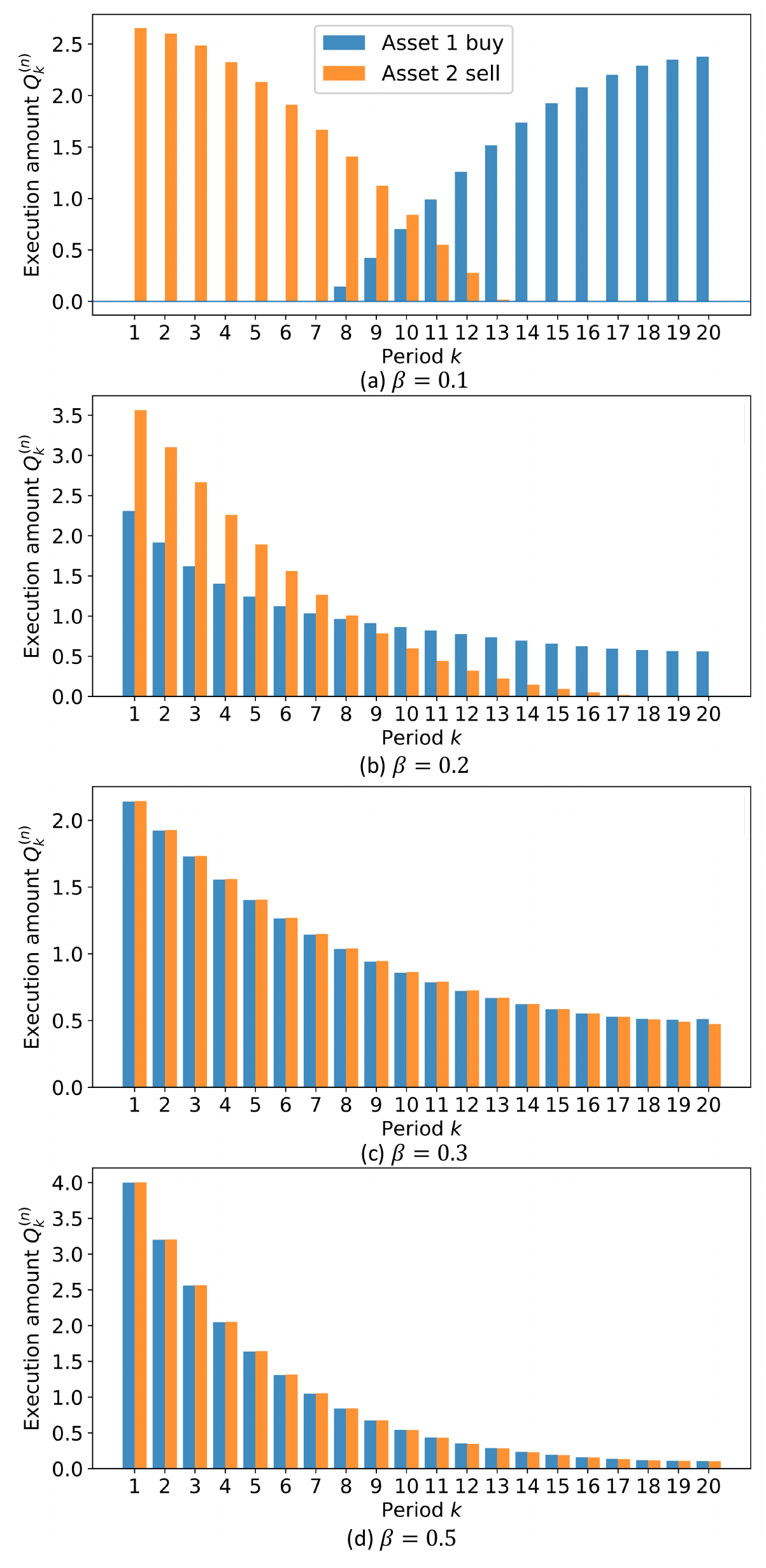}
    \caption{
    OE schedules under different cash constraint levels. The blue and orange bars represent the execution amounts of the buy and sell assets, respectively, at each execution period. The cash constraint is defined as $\bar{c}_k=\bar{c}_0+\beta k$, where smaller values of $\beta$ correspond to stricter cash constraints.
    }
    \label{Fig:exp1_result}
  \Description{OE schedules under different cash constraint levels.}
\end{figure}

\begin{table}[t]
\centering
\caption{Simulation settings for Experiment 2.}
\label{Tab:exp2_setting}
\small
\setlength{\tabcolsep}{3pt}
\begin{tabular}{lll}
\toprule
\textbf{Param.} & \textbf{Description} & \textbf{Value} \\
\midrule
$N$ & \# of assets & 2\\
$(Q^{(1)},Q^{(2)})$ & Order sizes ($\times10^{6}$) & \begin{tabular}[t]{@{}l@{}}
(0.5,-0.5), (0.75,-0.75),\\
(1.0,-1.0), (1.5,-1.5),\\
and (2.0,-2.0)
\end{tabular}\\
$T$ & Execution horizon & 7,200 sec.\\
$K$ & Number of execution periods & 20 \\
$\tau$ & Time interval & 360 sec. \\
$\lambda$ & Risk-aversion parameter & 0.50 \\
$\bar{c}_k$ & Cash budget for Ours & 0 for all $k$ \\
- & Simulator & Simudyne Pulse v2.0 \\
- & Calibration date & 2025-09-15 \\
- & Validation date & 2025-10-14 \\
- & Assets & \begin{tabular}[t]{@{}l@{}}
1024.HK, 9888.HK,\\
9999.HK
\end{tabular}\\
- & \# of buy--sell pairs & $3\times 2=6$\\
- & \# of validation runs & 300\\
Solver & Optimization solver & SCS \\
$V_i^{(n)}$ &\begin{tabular}[t]{@{}l@{}}
Target volume of TWAP simulations\\
for the parameter estimations 
\end{tabular} & \begin{tabular}[t]{@{}l@{}}1,000, 5,000, 10,000,\\
and 20,000 shares\end{tabular}\\
- & \# of calibration runs & $3$ for all $V_i^{(n)}$\\
\bottomrule
\end{tabular}
\end{table}

\section{Results}


\subsection{RQ1: Effect of Cash Constraints on OE}
We first investigate how the proposed cash constraints influence the structure of the OE schedule. Figure~\ref{Fig:exp1_result} illustrates the OE schedules obtained under different levels of cash constraints. When the cash constraint is relatively loose ($\beta=0.5$), the proposed method produces nearly symmetric buy and sell schedules that closely resemble the unconstrained AC solution. As the cash constraint becomes tighter, however, the execution schedules become increasingly asymmetric. In particular, buy orders are progressively postponed, while sell orders are executed earlier in the trading horizon. Under the strictest constraint ($\beta=0.1$), most sell orders are completed before substantial buy orders are initiated. This behavior arises because the optimization explicitly enforces the cash constraint at every trading period. Executing sell orders earlier generates cash inflows that, in expectation, finance subsequent buy orders, allowing the optimization to keep the expected cash process within budget throughout the execution horizon. Consequently, the proposed formulation automatically coordinates the timing of buy and sell executions according to the available cash budget, rather than simply optimizing each asset independently. These observations answer \textbf{RQ1}, demonstrating that explicit cash constraints fundamentally reshape the structure of the OE schedule by coordinating cash inflows and outflows across multiple assets.

\subsection{RQ2: Performance in Realistic Markets}
Table~\ref{Tab:estimated_parameters} summarizes the execution-model parameters estimated from the market simulator on the calibration date (2025-09-15). These parameters are subsequently fixed when constructing the OE schedules. As shown in Table~\ref{Tab:estimated_parameters}, the estimated market impact and volatility parameters vary substantially across assets, representing different liquidity conditions. In particular, 9999.HK exhibits the largest estimated market impact coefficient, indicating a relatively illiquid market in which trades exert a stronger influence on prices. Figure~\ref{Fig:simulated_prices} illustrates the price trajectories generated by calibrated Simudyne Pulse market simulator on the validation date (2025-10-14). It illustrates that the market simulator generates diverse price trajectories for each asset. Together, these results indicate that the proposed method is evaluated across heterogeneous market conditions with varying levels of liquidity and price volatility. 
\begin{figure}[tbp]
    \centering
    \includegraphics[width=0.9\linewidth]{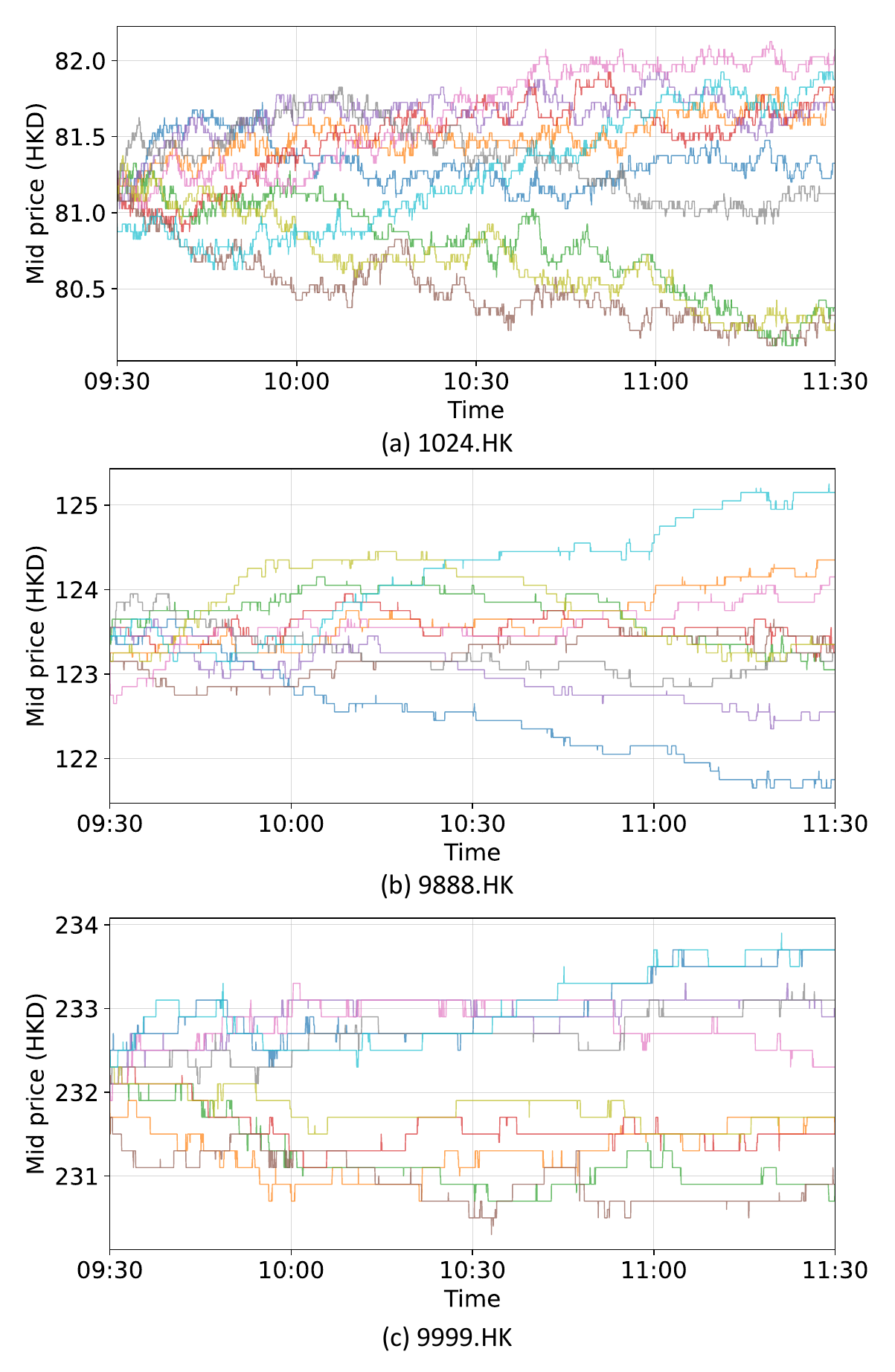}
    \caption{
    Representative simulated mid-price trajectories generated by the calibrated Simudyne Pulse market simulator on the validation date (2025-10-14). Each line represents one Monte Carlo realization of the same market environment. These simulated price paths are used exclusively for out-of-sample evaluation of execution strategies.
    }
    \label{Fig:simulated_prices}
    \Description{Representative simulated mid-price trajectories generated by the calibrated Simudyne Pulse market simulator on the validation date (2025-10-14).}
  \end{figure}

\begin{table}[t]
\centering
\caption{Market parameters estimated from the calibration simulation on 2025-09-15. The arrival price $P_0^{(n)}$, permanent impact coefficient $\gamma^{(n)}$, temporary impact coefficient $\eta^{(n)}$, and volatility $\sigma^{(n)}$ are estimated independently for each asset and subsequently used to construct the execution schedules evaluated on the validation date.}
\label{Tab:estimated_parameters}
\begin{tabular}{lcccc}
\toprule
\textbf{Symbol} &
$P_0^{(n)}$ &
$\gamma^{(n)}$ &
$\eta^{(n)}$ &
$\sigma^{(n)}$ \\
\midrule
1024.HK & 81.275  & $1.523\times10^{-9}$  & $6.524\times10^{-4}$ & $1.666\times10^{-3}$ \\
9888.HK & 123.350 & $8.650\times10^{-11}$ & $0.000$              & $2.529\times10^{-3}$ \\
9999.HK & 232.500 & $1.899\times10^{-8}$  & $1.038\times10^{-2}$ & $1.622\times10^{-3}$ \\
\bottomrule
\end{tabular}
\end{table}

\begin{table*}[t]
\centering
\caption{
Performance comparison of execution strategies under realistic market simulations ($Q^{(1)}=0.5\times10^{6},Q^{(2)}=-0.5\times10^6$). Results are reported as mean with standard deviation in parentheses over 300 independent simulations. Peak Cash Drawdown measures the largest cumulative cash shortfall during execution, normalized by the buy notional. Combined IS denotes the IS of the buy--sell pair. Lower values indicate better performance for both metrics.
}
\label{Tab:performance}
\setlength{\tabcolsep}{4pt}
\renewcommand{\arraystretch}{1.15}
\begin{tabular}{lcccccccc}
\toprule
&
\multicolumn{4}{c}{\textbf{Peak Cash Drawdown (\%)}} &
\multicolumn{4}{c}{\textbf{Combined IS (bps)}} \\
\cmidrule(lr){2-5}
\cmidrule(l){6-9}
\textbf{Pair (Buy/Sell)}
& \textbf{TWAP}
& \textbf{VWAP}
& \textbf{AC free}
& \textbf{Ours}
& \textbf{TWAP}
& \textbf{VWAP}
& \textbf{AC free}
& \textbf{Ours} \\
\midrule
\small{1024.HK / 9888.HK}
& \textbf{0.00 (0.00)}
& \textbf{0.00 (0.00)}
& 59.73 (2.53)
& 1.20 (1.91)
& -13.08 (57.61)
& \textbf{-13.23 (46.88)}
& -13.06 (42.15)
& -13.20 (57.34) \\

\small{1024.HK / 9999.HK}
& \textbf{0.00 (0.00)}
& \textbf{0.00 (0.00)}
& 88.45 (4.15)
& 0.67 (8.15)
& 17.39 (79.75)
& 20.41 (68.23)
& 8.52 (70.29)
& \textbf{8.23 (80.28)} \\

\small{9888.HK / 1024.HK}
& 10.16 (0.01)
& 10.99 (0.07)
& 60.27 (2.36)
& \textbf{2.01 (2.22)}
& \textbf{9.26 (58.34)}
& 10.61 (47.54)
& 12.32 (48.13)
& 10.13 (56.69) \\

\small{9888.HK / 9999.HK}
& \textbf{0.00 (0.00)}
& 1.11 (0.07)
& 57.06 (2.53)
& 0.03 (0.27)
& 17.72 (77.74)
& 27.71 (62.10)
& 13.90 (70.93)
& \textbf{11.77 (77.02)} \\

\small{9999.HK / 1024.HK}
& 16.67 (1.03)
& 19.52 (1.00)
& 70.03 (3.52)
& \textbf{11.81 (4.06)}
& -20.39 (78.03)
& \textbf{-27.47 (66.01)}
& -8.44 (53.28)
& -11.34 (78.23) \\

\small{9999.HK / 9888.HK}
& 10.51 (1.03)
& 10.80 (1.08)
& 69.43 (3.30)
& \textbf{9.98 (3.99)}
& -41.98 (75.89)
& \textbf{-46.92 (60.85)}
& -30.84 (50.82)
& -33.74 (75.58) \\
\bottomrule
\end{tabular}
\end{table*}

Table~\ref{Tab:performance} compares all execution strategies at the benchmark target execution volume, $Q^{(1)}=0.5\times10^{6}$ and $Q^{(2)}=-0.5\times10^{6}$. Across all six ordered asset pairs, our method substantially reduces peak cash drawdown relative to AC free. While AC free incurs drawdowns ranging from 57.06\% to 88.45\%, the corresponding values for our method remain between 0.03\% and 11.81\%. The reduction is particularly pronounced for the pairs involving 9999.HK, for which AC free requires the largest intermediate cash balances. This observation is consistent with the relatively large estimated market-impact coefficient and low liquidity of 9999.HK reported in Table~\ref{Tab:estimated_parameters}. These results show that incorporating the cash constraint directly into the optimization problem effectively limits intermediate funding requirements during execution.

Importantly, the improvement in cash feasibility does not come at the expense of execution quality. At the benchmark volume, our method achieves a lower mean combined IS than AC free for every evaluated asset pair. The differences are relatively small for some pairs, but they consistently favor our method despite the substantial reduction in cash drawdown. TWAP and VWAP attain lower combined IS for some asset pairs, whereas our method performs best for others. Thus, our method does not uniformly dominate the rule-based baselines in terms of IS alone; rather, it provides a substantially improved balance between execution cost and cash feasibility through an explicit optimization framework.

\begin{figure}[t]
    \centering
    \includegraphics[width=\linewidth]{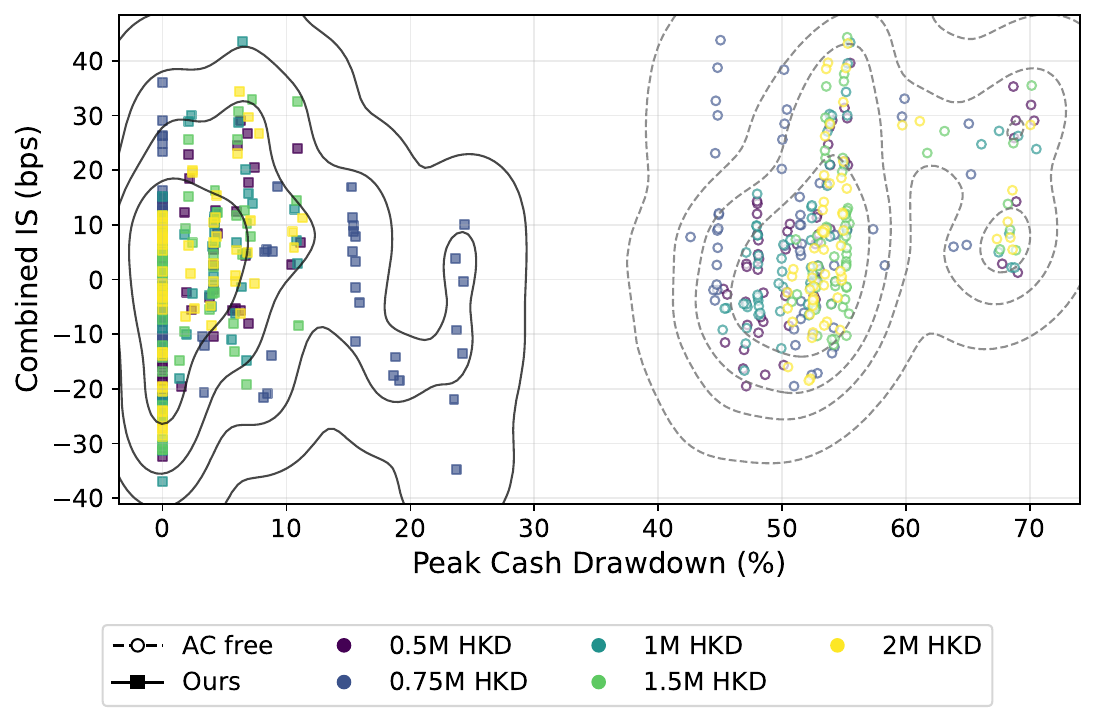}
    \caption{
    Joint distributions of peak cash drawdown and combined IS for Ours and AC free. Solid and dashed contours show KDEs for Ours and AC free, respectively, and markers indicate simulation outcomes. Results are aggregated over five execution volumes, six ordered asset pairs, and ten runs per setting (300 observations per method). Across all paired observations, the difference in combined IS (Ours $-$ AC free) is $-6.68$ bps on average (standard deviation: $12.95$).
    }
    \label{Fig:kde_cash_drawdown_combined_is}
    \Description{Joint distributions of peak cash drawdown and combined IS for Our method and AC free.}
\end{figure}

Figure~\ref{Fig:kde_cash_drawdown_combined_is} further demonstrates that Ours remains effective across different execution scales. The distribution of Ours is clearly shifted toward lower peak cash drawdown than that of AC free, and the two distributions are largely separated along the cash-drawdown dimension. This indicates that the reduction in intermediate cash requirements is robust across different execution volumes and asset pairs. In contrast, the combined-IS distributions partially overlap, reflecting variability across market conditions and simulation runs. Nevertheless, the proposed method achieves a lower combined IS on average while substantially reducing peak cash drawdown. These results suggest that explicitly incorporating cash constraints not only improves cash feasibility but also enhances execution performance over a wide range of execution conditions.

One possible interpretation is that the cash constraint acts as an implicit regularizer on the execution schedule. By preventing excessively cash-intensive schedules, it may also suppress aggressive trades that generate large market impact under the simulated market dynamics. Although the present results do not directly identify this mechanism, they show that cash feasibility and execution efficiency need not constitute a strict trade-off. Overall, these findings answer \textbf{RQ2}: explicitly incorporating cash constraints substantially reduces intermediate cash consumption while preserving, and on average improving, execution efficiency.

\section{Discussion}

\paragraph{\textbf{Execution Feasibility as a New Perspective on OE}} The classical OE literature~\citep{ac,multi_asset_oe_cross_impact1,multi_asset_oe_covariance1} focuses on minimizing IS while balancing market impact and execution risk, implicitly assuming sufficient capital throughout execution. Our results show that this assumption can be problematic in portfolio rebalancing with simultaneous buy and sell orders, where temporary cash shortages may arise even when the target portfolio is financially feasible.

The proposed formulation extends the conventional notion of OE by explicitly incorporating cash feasibility in expectation, constraining the expected cumulative cash consumption at every period of the execution horizon. Rather than solely minimizing execution costs, the optimization simultaneously determines an execution schedule whose expected capital deployment remains within budget throughout the trading horizon. The experimental results demonstrate that this additional consideration substantially reduces intermediate funding requirements while preserving execution efficiency. These findings suggest that execution feasibility should be regarded as an important optimization objective alongside transaction cost and risk in practical multi-asset execution.

\paragraph{\textbf{Implications for Practical Portfolio Execution}} An important property of the proposed formulation is that explicit cash constraints can be incorporated without sacrificing computational tractability. Although intertemporal cash constraints introduce temporal dependencies across assets through cumulative cash flows, the resulting optimization remains a convex QCQP under the standard linear market impact assumptions. Consequently, globally optimal schedules can be obtained using existing convex solvers without heuristic or black-box methods.

This analytical formulation also offers practical advantages over purely data-driven approaches~\citep{oe_w_cash_constraint1}. Since the cash constraints are modeled explicitly, practitioners can directly specify available capital limits according to operational requirements while preserving the interpretability of the resulting execution schedules. Such transparency is particularly desirable in institutional portfolio rebalancing, where execution decisions must often satisfy predefined funding and risk management policies.

\paragraph{\textbf{Limitations and Future Work}}

Several limitations remain to be addressed. First, the proposed framework adopts the classical linear AC market impact model. Although this assumption enables a convex analytical formulation, real financial markets exhibit nonlinear impact~\citep{nonlinear_impact1,nonlinear_impact2,nonlinear_impact3}. Nevertheless, we view the proposed framework as a theoretical benchmark. By providing the first convex formulation of multi-asset execution with intertemporal cash constraints, it establishes a reference point for future studies that extend the framework to more realistic market impact models. Second, the current formulation constrains the expected cumulative cash process rather than realized cash trajectories. Under stochastic price fluctuations, the realized cash consumption may temporarily exceed the prescribed budget even when the expectation satisfies the constraint. Incorporating probabilistic or robust cash constraints represents another promising extension. Finally, our empirical evaluation focuses on two-asset portfolio rebalancing tasks using parameters calibrated from an agent-based market simulator. Although the optimization framework is applicable to an arbitrary number of assets, further validation on larger portfolios would provide stronger evidence of its practical effectiveness. 

\section{Conclusion}
This paper introduced a cash-constrained formulation of multi-asset OE by extending the classical AC framework with intertemporal constraints on expected cash consumption. We showed that the resulting problem can be reformulated as a convex QCQP, allowing globally optimal schedules to be computed efficiently while constraining expected capital deployment throughout execution. Synthetic experiments showed that the proposed constraints fundamentally reshape execution schedules by coordinating cash inflows and outflows across assets. Furthermore, out-of-sample evaluations using an agent-based market simulator showed that the proposed method substantially reduces peak cash drawdown while maintaining comparable IS under various market conditions. Our results suggest that execution feasibility should be regarded as a fundamental objective alongside execution cost and risk. We hope the proposed formulation serves as a benchmark for future studies on feasibility-aware OE, including richer market impact models, stochastic cash constraints, and learning-based execution.

\bibliographystyle{ACM-Reference-Format}
\bibliography{bib} 


\end{document}